\documentclass{article}

\usepackage{arxiv}

\usepackage[utf8]{inputenc} 
\usepackage[T1]{fontenc}    
\usepackage{hyperref}       
\usepackage{url}            
\usepackage{booktabs}       
\usepackage{amsfonts}       
\usepackage{nicefrac}       
\usepackage{microtype}      
\usepackage{lipsum}		
\usepackage{graphicx}
\usepackage{doi}
\usepackage{bm}
\usepackage{amsthm}
\usepackage{amsmath}
\theoremstyle{plain}
\newtheorem{theorem}{Theorem}

\newtheorem{corollary}{Corollary}
\newtheorem{proposition}[theorem]{Proposition}

\theoremstyle{definition}
\newtheorem{definition}{Definition}

\theoremstyle{remark}
\newtheorem{remark}{Remark}

\title{Maximum entropy probability distributions on spheres with fixed mean Busemann function and holomorphic-information-geometric model of cognition}

\author{Vladimir Ja\' cimovi\' c \\
	Faculty of Natural Sciences and Mathematics\\
	University of Montenegro\\
	Cetinjski put bb., 81000 Podgorica\\
	Montenegro\\
	\texttt{vladimirj@ucg.ac.me} \\
}

\renewcommand{\shorttitle}{MaxEnt on spheres}

\hypersetup{
pdftitle={Maximum entropy probability distributions on spheres with fixed mean Busemann function and holomorphic-information-geometric model of cognition},
pdfsubject={-},
pdfauthor={Vladimir Ja\' cimovi\' c},
pdfkeywords={Bergman ball, Poincar\' e ball, Free energy, RKHS},
}

\begin{document}
\maketitle

\bigskip

{\bf Abstract}

In the first half of the paper we revisit the question regarding MaxEnt probability distributions on spheres. We derive families of MaxEnt distributions on spheres in real and complex vector spaces with fixed expected Busemann function (energy). As a particular case, we deduce sub-families on canonical energy levels where inverse temperature equals the dimension of the sphere. In the second part we focus on the information manifold of canonical MaxEnt distributions on the sphere in the complex vector space. This manifold is isomorphic to the Bergman ball. We introduce the reproducing kernel on this manifold and use the RKHS theory to elaborate the model of cognition. In particular, we state the principle of minimal cognitive effort in RKHS and demonstrate its dual relationship with the MaxEnt principle for probability distributions on the boundary sphere.

\bigskip

\section{Introduction}\label{sec:1}

The Principle of Maximum Entropy (MaxEnt) is one of the most universal principles in science. It was the central principle of thermodynamics and statistical mechanics since their beginnings in the 19th century, but its reach extends far beyond classical physics. Most notably, MaxEnt is the foundation of cognition and reasoning. This interpretation was formalized in the mid-20th century when Jaynes reimagined entropy as a measure of the mind projection fallacy \cite{Jaynes}. This led to Objective Bayesianism and Probability as Extended Logic. In this hypostasis, the MaxEnt principle serves as a computational tool for logical inference, asserting that the most objective probability distribution is the one that maximizes uncertainty while satisfying known constraints. Although various manifestations of the MaxEnt principle in statistics are well known and established as rigorous theorems, they remain frequently revisited. This continuous reexamination highlights an ever-lasting question: how varying interpretations of missing information and physical constraints reshape our understanding of complex systems. Consequently, exploring MaxEnt distributions is not a solved mathematical problem, but an evolving framework for updating hypotheses about the world from minimal assumptions \cite{Entropy}.

The notion of free energy serves as a bridge between microscopic states and macroscopic system behaviors in statistical mechanics. While the MaxEnt principle in Objective Bayesianism highlights the most objective distributions under general constraints, introducing energetic bounds yields Boltzmann-Gibbs distributions as equilibrium states of a system. These distributions formalize the fundamental balance between minimizing internal energy and maximizing entropy. Furthermore, the generalization of this concept into variational free energy positions this statistical-mechanical framework as one of the foundational paradigms of modern machine learning and inference \cite{JGJS,Friston}.

The concepts of entropy and free energy can be extended to general metric spaces. This includes positively curved Riemannian manifolds, where probability distributions are studied within the sub-field known as directional statistics \cite{MJ}. The simplest Riemannian manifold is the circle. Hence, a good starting point for our exposition is the following question: What are the MaxEnt distributions on the circle $\mathbb{S}^1$? 

The routine answer is that these are von Mises distributions, defined by the following density functions 
\begin{equation}
\label{von Mises}
p_{vM}(\varphi;\mu,\kappa) = \frac{1}{2 \pi I_0(\kappa)} exp\{\kappa \cos(\varphi - \mu)\}, \quad \varphi \in [0,2 \pi),
\end{equation}
where $\mu \in \mathbb{S}^1, \, \kappa>0$ are parameters and $I_0(\cdot)$ denotes the first order Bessel function.

  More precise, mathematically rigorous assertion is that von Mises densities \eqref{von Mises} maximize entropy for fixed first directional moment (the centroid). 
  
  One could extend to higher dimensions: What are the MaxEnt distributions on spheres? Analogously, the routine answer would be that these are von Mises-Fisher distributions defined by the following densities on spheres $\mathbb{S}^{d-1}$ \cite{MJ}:
\begin{equation}
\label{von Mises-Fisher}
p_{vMF}(x;\mu,\kappa) = C_d(\kappa) exp \{ \kappa \mu^T x\}, \quad x \in \mathbb{S}^{d-1},
\end{equation}
where $\mu \in \mathbb{S}^{d-1}, \, \kappa>0$ and $C_d$ is the normalizing constant. The notion $\mu^T$ stands for the transpose vector of $\mu$.

However, the notion of a mean is ambiguous on spheres and balls. This ambiguity further extends to the notion of energy. Consequently, alternative definitions of a mean may be more appropriate in many physical or biological setups. This ultimately boils down to the question of which geometry we assume in the interior balls. Very often, Euclidean metric, which turns balls into compact manifolds, is not physically justified assumption.

One notion of the mean of a probability measure on the sphere is conformal barycenter introduced by Douady and Earle in their seminal paper \cite{DE}. This barycenter is associated with the hyperbolic metric in the ball. It is defined as a unique (under mild conditions) stationary point of the expected Busemann function. 

The entire present study relies on the notion of the Busemann function \cite{Busemann}. This function acts as a directional potential field in the homogeneous manifold, determining orientations of paths, distances, and distributions aligned with the negative curvature. Hence, we will also use the terms Busemann energy or Busemann distance (although it is not the distance in the mathematically strict sense) for this function.

\subsection{Outline}

In Section 2 we introduce new families of MaxEnt distributions on spheres with fixed expected Busemann distance to an interior reference point. We start our analysis with the circle and its interior Poincar\' e disc. When extending to higher dimensions there are two non-equivalent generalizations, corresponding to two models of hyperbolic balls: Poincar\' e and Bergman. We solve isoperimetric problems to derive families of MaxEnt distributions for both models. Remarkably, normalizing constants can be computed explicitly for both models in all dimensions.

Our further exposition focuses on spheres in $\mathbb{C}^m$ with the interior balls which we equip with the Bergman metric. We introduce the Busemann free energy and solve the corresponding variational problem to rederive our MaxEnt family as Boltzmann-Gibbs distributions on the sphere in $\mathbb{C}^m$. We emphasize the canonical temperature scale, with the Lagrange multiplier equal to the complex dimension of the ball. In this case some parameters get eliminated and Boltzmann-Gibbs densities obtain a particularly elegant form. This reveals that on this canonical energy level the entire geometry of Bergman balls arises as a consequence of the free energy minimization. We continue by introducing a reproducing kernel on Bergman balls as a similarity measure between two canonical Boltzmann-Gibbs distributions.

 In the second part of the paper we apply MaxEnt distributions on spheres in $\mathbb{C}^m$ and reproducing kernel Hilbert spaces (RKHSs) in Bergman balls $\mathbb{\mathbb{B}}_{\mathbb{C}}^m$ to formalize a holomorphic information-geometric model of (collective) cognition. The proposed framework exhibits high expressive capacity, enabling it to capture subtle contextual effects and highly non-trivial patterns. The underlying cause is semantic richness of the Bergman ball, manifested in beautiful mathematical properties, such as large symmetry group $SU(m,1)$ of unitary Lorentz transformations, anisotropic metric, multifaceted equivariant kernels, and expressive RKHSs. This provides a compelling paradigm for learning complex cognitive phenomena and patterns that traditional Euclidean models fail to capture. 

In particular, by invoking the representer theorem for RKHSs \cite{HSS}, we establish the principle of minimal effort in cognition and relate it to our family of Boltzmann-Gibbs distributions on spheres in complex vector spaces. This connection suggests that cognitive processing can be modeled as a variational problem in the Bergman ball. By grounding optimization of the cognitive uncertainty in the language of statistical mechanics in Bergman balls, we develop a unified mathematical framework where cognitive economy and geometric constraints are two sides of the same coin. Since the Bergman ball appears in quantum theory as the space of $SU(m,1)$-coherent states, the interpretations in terms of quantum cognition and quantum decision-making arise naturally. We conclude the paper with the discussion on interpretations, strengths and limitations of the proposed model, followed by an outlook on possible extensions and potential applications in AI and deep learning.

\section{MaxEnt probability distributions on spheres on Busemann energy levels}

In this Section we expose a hyperbolic-geometric point of view on MaxEnt probability distributions on spheres. We stick to the standard Shannon differential entropy:
\begin{equation}
\label{entropy}
H(p) = - \int_{Sphere} p(x) \log p(x) d \sigma(x),
\end{equation}
where $p(x)$ is the density of a probability distribution and $\sigma(\cdot)$ is the metric element on the sphere.

We add the normalization constraint
\begin{equation}
\label{normalize}
\int_{Sphere} p(x) d \sigma(x) = 1.
\end{equation}

Throughout the paper we will assume that the measure on the sphere is normalized, so that
$$
\int_{Sphere} d \sigma(x) = 1.
$$
The exception is the circle, where we will use angular notations and keep the multiplier $1/(2 \pi)$. Normalization of the measure does not affect maxima of the entropy or minima of free energy. On the other hand, it introduces a constant (dimension dependent) multiplier into expression for the Fisher metric.

Of course, maximization of \eqref{entropy} under the constraint \eqref{normalize}, yields the uniform distribution on the sphere. Throughout this Section we impose an additional constraint by picking a point in the interior ball and fixing the mean Busemann distance to this reference point.

\subsection{MaxEnt probability distributions on the circle with fixed mean Busemann distance}

We start our analysis with the circle $\mathbb{S}^1$. Let $w = r e^{i \Phi}$ where $0 \leq r < 1$ be a point in the interior disc. Fix the mean value of the Busemann distance to the reference point $w$. Using the expression for the Busemann function in the disc, this constraint is written as
\begin{equation}
\label{Busemann_constraint_circle}
\int \limits_{0}^{2 \pi} p(\varphi) \log \frac{1-r^2}{1+r^2-2 r \cos (\varphi - \Phi)} d \varphi = const
\end{equation}

Consider the isoperimetric problem of maximizing the entropy \eqref{entropy} subject to constraints \eqref{normalize} and \eqref{Busemann_constraint_circle}. The Euler-Lagrange conditions for this problem yield
$$
\frac{\partial \cal L}{\partial p} = - \ln p - 1 - \alpha + \lambda \ln \frac{1-r^2}{1+r^2-2 r \cos (\varphi - \Phi)} = 0.
$$
Hence,
$$
\ln p(\varphi) = - (1 + \alpha) + \ln \left( \frac{1-r^2}{1+r^2-2r \cos(\varphi-\Phi)} \right)^\lambda
$$
and 
$$
p(\varphi) = e^{-(1 + \alpha)} \left( \frac{1-r^2}{1+r^2-2r \cos(\varphi-\Phi)} \right)^\lambda.
$$
By introducing the normalization constant $Z(\lambda,r^2)$ we rewrite the MaxEnt density as
$$
p(\varphi) = \frac{1}{Z(\lambda,r)} \frac{1}{(1+r^2-2r \cos(\varphi - \Phi))^\lambda}.
$$
By integrating $p(\varphi)$ on $\varphi \in [0,2 \pi)$ we find an explicit expression for the normalization constant:
\begin{equation}
\label{MaxEnt_Busemann_circle}
p(\varphi) = \frac{1}{2 \pi \; _2 F_1(\lambda,\lambda;1;r^2)}
 \frac{1}{(1+r^2-2r \cos(\varphi - \Phi))^\lambda}.
\end{equation}

This can be substantiated into the following
\begin{proposition}
MaxEnt distributions on the circle with fixed value of expected Busemann distance to the point $r e^{i \Phi}$ are given by densities \eqref{MaxEnt_Busemann_circle}.
\end{proposition}

Hence, the resulting MaxEnt family \eqref{MaxEnt_Busemann_circle} has a very natural physical interpretation in terms of energy or information in the Poincar\' e disc. One might ask what is an appropriate name for this family of distributions. It turns out that they already have the name. This family has been introduced in \cite{JP} as a flexible model for circular data and is known in directional statistics as the Jones-Pewsey family.

\subsubsection{MaxEnt probability distributions on the circle on canonical Busemann energy levels}

An exceptional (canonical) sub-family of \eqref{MaxEnt_Busemann_circle} arises if the energy level is fixed at $const = - \log(1-r^2)$. This corresponds to the Lagrange multiplier $\lambda=1$. In this case hypergeometric series in \eqref{MaxEnt_Busemann_circle} have a particularly simple form:
$$
_2 F_1(1,1;1,r^2) = \frac{1}{1-r^2}.
$$ 
Substituting this into \eqref{MaxEnt_Busemann_circle} yields
\begin{equation}
\label{wrapped_Cauchy}
p_{wC}(\varphi) = \frac{1}{2 \pi} \frac{1-r^2}{1-2 r \cos(\varphi-\Phi) + r^2}.
\end{equation}
Densities \eqref{wrapped_Cauchy} are classical mathematical objects known as Poisson kernels. In directional statistics, they define the family of wrapped Cauchy distributions on the circle \cite{McCullagh}. These densities are obtained by "wrapping" Cauchy densities on the real line via Cayley transformation.

\begin{corollary}
The wrapped Cauchy distributions \eqref{wrapped_Cauchy} are MaxEnt family on the circle with the energy constraint \eqref{Busemann_constraint_circle} with $const = - \log(1-r^2)$. This energy level corresponds to unit temperature. 
\end{corollary}

One of geometric harmonies of this specific case is that $w=r e^{i \Phi}$ is the conformal barycenter of the distribution \eqref{wrapped_Cauchy}.

\begin{remark}
For another, although far less prominent, sub-family of  \eqref{MaxEnt_Busemann_circle} we set $\lambda=-1$. This yields densities
$$
p_{cardio}(\varphi;\mu,\rho) = \frac{1}{2 \pi} (1 + 2 \rho \cos(\varphi-\mu))
$$ 
which are known in directional statistics as cardioid distributions \cite{MJ}, parametrized by $\mu \in [0,2 \pi)$ and $\rho \in [0,0.5]$.
\end{remark}

The above analysis can be extended to higher dimensions. However, at this point two directions should be distinguished. 

The first option is to consider the sphere $\mathbb{S}^{d-1}$ in the real vector space $\mathbb{R}^d$ and equip its interior ball with the Poincar\' e metric. The second option is to consider the sphere in the complex vector space $\mathbb{C}^m$ and equip the interior ball with the Bergman metric. Then we obtain the Bergman ball $\mathbb{B}_{\mathbb{C}}^m$ with the boundary sphere $\partial \mathbb{B}_{\mathbb{C}}^m$.

The two geometries are equivalent only in the case when $d=2m=2$, when both of them reduce to the Poincar\' e disc.

\subsection{MaxEnt distributions on spheres in real vector spaces with fixed mean Busemann function}

We start with the MaxEnt family on the spheres in real vector spaces with interior Poincar\' e balls. First, we introduce the manifold.

\subsubsection{Poincar\' e balls in real vector spaces and their boundary spheres}

Consider the unit ball in $d$-dimensional vector space:
$$
\mathbb{B}^d = \{ x = (x_1,\dots,x_d) \in \mathbb{R}^d \, : \, \| x \| < 1\} 
$$
where $\|x\| = \sqrt{x_1^2+\cdots+x_d^2}$ denotes the norm of the vector $x$.

Introduce the metric tensor at any point $x \in \mathbb{B}^d$ via the conformal scaling of standard Euclidean metric $g_E$:
\begin{equation}
\label{Poincare_metric}
g_x = \frac{4 g_E}{(1-\|x\|^2)^2}.
\end{equation}
Hence, the metric element $ds^2$ reads:
$$
ds^2 = \frac{4}{(1-\|x\|^2)^2} \sum_{i=1}^d dx_i^2.
$$
\begin{definition}
The Riemannian manifold $\mathbb{B}^d$ equipped with the metric \eqref{Poincare_metric} is called the $d$-dimensional Poincar\' e ball. 
\end{definition}

One can verify that the geodesic distance between any two points $u,v \in \mathbb{B}^d$ is
$$
d_{Poin}(u,v) = arcosh \left( 1 + 2 \frac{\|u-v\|^2}{(1-\|u\|^2)(1-\|v\|^2)} \right).
$$
Orientation-preserving isometry transformations of the Poincar\' e ball are given by
\begin{equation}
\label{Poincare_isometry}
h_b(x) =Q \frac{b \|x-b\|^2 + (1-\|b\|^2)(b-x)}{\rho(x,b)}
\end{equation}
where $Q$ is an orthogonal transformation of $\mathbb{R}^d$ and $\rho(\cdot,\cdot)$ is defined as
\begin{equation}
\label{rho}
\rho(x,b) = \|x-b\|^2 + (1-\|b\|^2)(1-\|x\|^2).
\end{equation}
The group of orientation-preserving isometries \eqref{Poincare_isometry} is isomorphic to the orthogonal Lorentz group $SO^+(d,1)$.

We point out the formula
$$
1 - \|h_b(x)\|^2 = \frac{(1-\|b\|^2)(1-\|x\|^2)}{\rho(x,b)},
$$
where $\rho(\cdot,\cdot)$ is defined by \eqref{rho}.

Jacobian of the mapping $w=h_b(x)$ is 
\begin{equation}
\label{Poincare_Jacobian}
J(x,w) = \frac{1-\|b\|^2}{\rho(x,b)^d} = \frac{(1-\|w\|^2)^d}{(1-\|x\|^2)^d}.
\end{equation}

Finally, denote by $\mathbb{S}^{d-1}$ the boundary sphere of the Poincar\' e ball:
$$
\mathbb{S}^{d-1} = \{ y \in \mathbb{R}^d \, : \, \|y\| = 1 \}.
$$

\subsubsection{The Busemann energy in Poincar\' e balls and MaxEnt distributions on their boundary spheres}

For $y \in \mathbb{S}^{d-1}$ and $w \in \mathbb{B}^d$ the Busemann function in the Poincar\' e ball reads:
\begin{equation}
\label{Busemann_Poincare}
\beta_y(w) = \log \frac{1-\|w\|^2}{\|y-w\|^2}, \quad \mbox{ where  } \;  \| y \| = 1, \, \|w\|<1.
\end{equation}
Impose the constraint on the expected Busemann distance to the reference point $w$:
\begin{equation}
\label{Busemann_constraint_Poincare}
\mathbb{E}[\beta_y(w)] \equiv \int_{\mathbb{S}^{d-1}} p(y) \log \frac{1-\|w\|^2}{\|y-w\|^2} d \sigma(y) = const.
\end{equation}

Consider the problem of maximizing entropy \eqref{entropy} under constraints \eqref{normalize} and \eqref{Busemann_constraint_Poincare}. 

The Euler-Lagrange condition reads
$$
\frac{\partial L}{\partial p}(p) = - \log p(y) - 1 - \alpha - \lambda \beta_y(w) = 0.
$$
By exponentiating the above equation we find that
$$
p(y) \propto e^{-\lambda \beta_y(w)}
$$
and substituting the Busemann energy \eqref{Busemann_Poincare}
$$
p(y) \propto \left( \frac{1-\|w\|^2}{\|y-w\|^2} \right)^\lambda.
$$
The normalizing constant is computed explicitly:
\begin{equation}
\label{normalizing_Poincare_sphere}
Z(\lambda,\|w\|^2) = \int_{\mathbb{S}^{d-1}} \left( \frac{1-\|w\|^2}{\|y-w\|^2} \right)^\lambda d \sigma(y) = (1-\|w\|^2)^\lambda \, _2F_1(\lambda,\lambda-\frac{d-2}{2};\frac{d}{2};\|w\|^2).
\end{equation}
We substantiate this derivation in the following
\begin{proposition}
MaxEnt probability distributions on spheres in real vector spaces with fixed mean Busemann distance \eqref{Busemann_Poincare} from the point $w \in \mathbb{B}^d$ are defined by the following densities
\begin{equation}
\label{MaxEnt_Busemann_Poincare}
p(y \, | \, \lambda,w) = \frac{1}{_2 F_1(\lambda,\lambda-\frac{d-2}{2};\frac{d}{2};\|w\|^2) \| y-w \|^{2 \lambda}}, \quad \| y \| =1.
\end{equation}
\end{proposition}
This family is parametrized by $w \in {\mathbb B}^d$ and $\lambda \in \mathbb{R}$. The Lagrange multiplier $\lambda$ determines energy levels.

\begin{remark}
As a simple but important remark we notice that the family \eqref{MaxEnt_Busemann_Poincare} includes the uniform distribution on $\mathbb{S}^{d-1}$ both for $w=0$ and for $\lambda = 0$ (corresponding to the energy level $const=0$).

On the other extreme, delta distributions at points on $\mathbb{S}^{d-1}$ arise as limit cases when $\| w \| \to 1$.
\end{remark}

\begin{remark}
For another interesting limit case, let $\lambda \to \infty$ and $\|w\| \to 0$, such that $\lambda \| w \| \to \kappa$.
\footnote{As $\lambda$ is interpreted as inverse temperature, this corresponds to the limit of zero temperature. In this case, the Busemann function linearizes into a flat dot product.} 
One can show that in this case the densities \eqref{MaxEnt_Busemann_Poincare} converge to the von Mises-Fisher density \eqref{von Mises-Fisher} with the concentration parameter $\kappa$. As pointed out in Introduction, this is precisely the MaxEnt distribution with the fixed Euclidean centroid.
\end{remark}

\subsubsection{MaxEnt probability distributions on spheres in real vector spaces on canonical energy levels}

A canonical case arises when the Lagrange multiplier is equal to the dimension of the sphere: $\lambda = d-1$.
In this case, the hypergeometric function in the normalizing constant \eqref{normalizing_Poincare_sphere} has a simple evaluation:
\begin{equation}
\label{hyp_Poin_d-1}
_2 F_1(d-1,\frac{d}{2};\frac{d}{2};\|w\|^2) = (1-\|w\|^2)^{-(d-1)}.
\end{equation}
Substituting this into \eqref{MaxEnt_Busemann_Poincare} we obtain the following expression for densities:
\begin{equation}
\label{spherical_Cauchy}
p(y \, | \, w) = \left( \frac{1-\|w\|^2}{\|y-w\|^2} \right)^{d-1}.
\end{equation}
The family \eqref{spherical_Cauchy} has been studied in \cite{KMcC}. We will name it {\it spherical Cauchy} family and denote by ${\cal SC}(w)$. These distributions are extensions of wrapped Cauchy distributions \eqref{wrapped_Cauchy} to higher dimensions with analogous beautiful properties. In particular, ${\cal SC}(w)$ is invariant w. r. to isometries \eqref{Poincare_isometry} of the Poincar\' e ball. More precisely, if $X \sim {\cal SC}(w)$ and $g$ is an isometry, then $Y = g(X) \sim {\cal SC}(g(w))$.  

In order to determine the energy level corresponding to this family, we differentiate the log-partition:
\begin{equation}
\label{energy_level_log_partition}
const = \mathbb{E}[\beta_y(w)] = - \frac{\partial}{\partial \lambda} \log Z(\lambda,\|w\|^2) =
- \log(1-\|w\|^2) - \frac{\frac{\partial}{\partial \lambda} \, _2F_1(\lambda,\lambda-\frac{d-2}{2};\frac{d}{2};\|w\|^2)}{_2 F_1(\lambda,\lambda-\frac{d-2}{2};\frac{d}{2};\|w\|^2)}.
\end{equation}
For $\lambda=d-1$ the hypergeometric function in the nominator can be differentiated to get an explicit expression (see Appendix for the derivation):
$$
\frac{\partial}{\partial \lambda} \left. _2F_1(\lambda,\lambda-\frac{d-2}{2};\frac{d}{2};\|w\|^2) \right |_{\lambda = d-1} = - \frac{2}{(1-\|w\|^2)^{d-1}} \log(1-\|w\|^2) + \frac{1}{(1-\|w\|^2)^{d-1}} P_{d-3}(\|w\|^2), 
$$
where $P_{d-3}$ is the polynomial given by
$$
P_{d-3}(x) = \sum_{j=1}^{d-2} \frac{1}{j} \sum_{k=1}^j (-1)^k \binom{j}{k} x^k.
$$
As the denominator is given by \eqref{hyp_Poin_d-1}, we can substitute into \eqref{energy_level_log_partition} and cancel out the terms to get the final formula for the canonical energy level:
\begin{equation}
\label{energy_level_log_partition1}
const = \log(1-\|w\|^2) - P_{d-3}(\|w\|^2).
\end{equation}

We substantiate the above considerations to introduce the canonical sub-family of MaxEnt distributions.

\begin{proposition}
The MaxEnt family of probability distributions on $d-1$-dimensional sphere in real vector space for the Busemann energy level given by \eqref{energy_level_log_partition1} is the spherical Cauchy family defined by densities \eqref{spherical_Cauchy}.
The corresponding value of the inverse temperature is $\lambda = d-1$.
\end{proposition}

We write down energy levels for several dimensions:
\begin{itemize}

\item{The disc ($d=2$). In this case $d-2=0$ and the sum for $P_{d-3}$ is empty. Hence, it equals zero and the canonical energy level is simply the logarithm: $const = \log(1-\|w\|^2)$, as already derived for wrapped Cauchy on the circle.}

\item{Three-dimensional ball. For $d=3$ the sum for $P_{d-3}$ contains only one term and the canonical energy level is $const = \log(1-\|w\|^2) + \|w\|^2$.} 

\item{$d=4$. In this case we have two additional terms and $const = \log(1-\|w\|^2) + 2 \|w\|^2 + \|w\|^4/2$.} 

\end{itemize}

The above formulae illustrate how canonical energy levels depend on the dimension. 
 
\begin{remark}
Fisher information for the family \eqref{MaxEnt_Busemann_Poincare} is given by
$$
F_{sC}(w) = \frac{(d-1)^2}{d} \frac{4}{(1-\|w\|^2)} I.
$$
We refer to \cite{KMcC,Jacimovic1} for the proof. 

This expression for the Fisher information shows that the information manifold defined by densities ${\cal SC}(w)$ is isomorphic to the Poincar\' e ball up to the dimension-dependent multiplier $(d-1)^2/d$.
\end{remark}

\subsection{MaxEnt distributions on spheres in complex vector spaces with fixed Busemann energy levels}

We proceed with an alternative geometric framework, which is conveniently described in complex vector notations. We will introduce the Bergman ball and its boundary sphere and deal with these manifolds in the remainder of the paper.

We consider the complex vector space $\mathbb{C}^m$. The notation $\bar a$ stands the conjugate of the complex number $a$, while $\xi^\dagger$ denotes the conjugate transpose of the vector $\xi$. Hence, the expression $\xi^\dagger z$ is the standard Hermitian scalar product of vectors $\xi,z \in \mathbb{C}^m$ and $\|z\| = \sqrt{z^\dagger z}$ denotes the norm of the vector $z$. Accordingly, $z z^\dagger$ is the rank one matrix.

\subsubsection{Bergman balls, their isometry groups and boundary spheres}

Consider the unit ball in $\mathbb{C}^m$:
$$
\mathbb{B}_{\mathbb{C}}^m = \{ z \; : \; \|z\| <1\}.
$$


The Bergman kernel in $\mathbb{B}_{\mathbb{C}}^m$ is \cite{Zhu}
\begin{equation}
\label{Bergman_kernel}
K(z,w)=\frac{1}{m+1}\frac{1}{(1-\left<z,w\right>)^{m+1}}.
\end{equation}
Equip $\mathbb{B}_{\mathbb{C}}^m$ with the metric tensor $
g_z(u,v) = (B(z)u)^\dagger v$ for $u, v\in \mathbb{C}^m$, $z\in \mathbb{B}^m_{\mathbb{C}}.$
Here $$B(z)=(b(z)_{ij})_{i,j=1}^m \quad \mbox{ and } \quad b(z)_{ij}= \frac{1}{m+1}\frac{\partial^2}{\partial \overline{{z_i}}\partial z_j}K(z,z).$$ 

\begin{definition}
The Riemannian manifold $\mathbb{B}^m_{\mathbb{C}}$ with the metric tensor $g$ is named the Bergman  ball.
\end{definition}
Bergman balls have constant negative sectional curvature, see \cite{HL}.

Let $P_a$ be the
orthogonal projection of $\mathbb{C}^m$ onto the subspace $[a]$ generated by $a$, and let $$Q=Q_a =
I - P_a$$ be the projection onto the orthogonal complement of $[a]$. Explicitly, $P_0 = 0$ and  $P=P_a(z) =\frac{(z^\dagger a) a}{\|a \|^2}$. Set $s_a = (1 - \|a\|^2)^{1/2}$ and consider the map
\begin{equation}
\label{Bergman_transf}
m_a(z) =\frac{a-P_a z-s_a Q_a z}{1-z^\dagger a}.
\end{equation}
Compositions of mappings of the form \eqref{Bergman_transf} and unitary linear mappings in $\mathbb{C}^m$ consitute the group of holomorphic automorphisms of the unit ball $\mathbb{B}^m_{\mathbb{C}}$. It is easy to verify that $m_a^{-1}=m_a$. Moreover, for any automorphism $q$ of the Bergman ball onto itself there exists a unitary transformation $U$ such that 
\begin{equation}
\label{automob} 
m_{q(c)}\circ q=U\circ m_c.
\end{equation}

By using the representation formula \cite[Proposition~1.21]{Zhu}, the Bergman metric can be introduced by the following expression
\begin{equation}
\label{bergmet}
d_{Bergman}(z,w)=\frac{1}{2}\log\frac{1+\|m_w(z)\|}{1-\|m_w(z)\|}.
\end{equation}

It is well-known that every automorphism $q$ of the unit ball is an isometry w.r. to the Bergman metric, that is: $d_{Bergman}(z,w)=d_{Bergman}(q(z),q(w))$. The group of orientation-preserving automorphisms of the Bergman ball is isomorphic to the upper sheet $SU^+(m,1)$ of the unitary Lorentz group.

We also point out the formula 
\begin{equation}
\label{phia}
1-\|m_a(z)\|^2 = \frac{(1-\|z\|^2)(1-\|a\|^2)}{|1- a^\dagger z|^2}\end{equation} 
and the expression for the Jacobian
$$
J(z,m_a)=\left(\frac{1-\|m_a(z)\|^2}{1-\|z\|^2}\right)^{m+1}=\left(\frac{1-\|a\|^2}{|1- z^\dagger a|^2}\right)^{m+1}.
$$

Finally, we will use the notation $\partial \mathbb{B}^m_{\mathbb C}$ for the boundary sphere of the Bergman ball.

\subsubsection{The Busemann energy in Bergman balls and MaxEnt distributions on their boundary spheres}

The Busemann function in the Bergman ball is introduced as follows
\begin{equation}
\label{Busemann_Bergman}
\beta_\xi(z) = \log \frac{1-\|z\|^2}{|1-z^\dagger \xi|^2}, \quad \| \xi \| = 1 \; \| z \| < 1.
\end{equation}
Accordingly, the constraint on the mean Busemann energy reads
\begin{equation}
\label{Busemann_constraint_Bergman}
\mathbb{E}[\beta_\xi(z)] \equiv \int_{\partial \mathbb{B}^m_{\mathbb C}} \log \frac{1-\|z\|^2}{|1-z^\dagger \xi|^2} p(\xi) d \sigma(\xi) = const.
\end{equation}

Now, we solve the problem of maximizing \eqref{entropy} subject to constraints \eqref{normalize} and \eqref{Busemann_constraint_Bergman} where all integrals are evaluated over the sphere $\partial \mathbb{B}_{\mathbb C}^m$.

Apply the Euler-Lagrange condition
$$
\frac{\partial {\cal L}}{\partial p} = 0 \implies p(\xi) \propto e^{-\lambda \beta_\xi(z)}.
$$
Plugging the Busemann function \eqref{Busemann_Bergman} the density is further rewritten as
$$
p(\xi \, | \, \lambda,z) \propto \left( \frac{1-\|z\|^2}{|1-z^\dagger \xi|^2} \right)^\lambda.
$$
The normalizing constant can be computed by evaluating the spherical integral:
\begin{equation}
\label{normalizing_constant_Bergman}
Z(\lambda,\|z\|^2) = \int_{\partial \mathbb{B}^m_{\mathbb C}} \left( \frac{1-\|z\|^2}{|1-z^\dagger \xi|^2} \right)^\lambda d \sigma(\xi) = (1-\|z\|^2)^\lambda \, _2 F_1(\lambda,\lambda;m;\|z\|^2).
\end{equation}
Putting all together we assert the following
\begin{proposition}
MaxEent probability distributions on the sphere $\partial \mathbb{B}_{\mathbb C}^m$ with the fixed mean Busemann distance from the reference point $z$ in the Bergman ball ${\mathbb B}_\mathbb{C}^m$ are defined by the following densities
\begin{equation}
\label{MaxEnt_Busemann_constraint_Bergman}
p(\xi \, | \, \lambda,z) = \frac{1}{_2 F_1(\lambda,\lambda;m;\|z\|^2) |1-z^\dagger \xi|^{2 \lambda}}.
\end{equation}
\end{proposition}

\subsubsection{MaxEnt probability distributions on spheres in complex vector spaces on canonical energy levels}

The exceptional case arises when the Lagrange multiplier is equal to the complex dimension, that is $\lambda = m$. In this case the hypergeometric function in \eqref{MaxEnt_Busemann_constraint_Bergman} has the following simple form:
\begin{equation}
\label{hyp_geom_eval}
_2 F_1(m,m;m;\|z\|^2) = (1-\|z\|^2)^{-m}.
\end{equation}
Substituting this into \eqref{MaxEnt_Busemann_constraint_Bergman} we find that
\begin{equation}
\label{Cauchy-Bergman}
p(\xi \, | \, z) = \left( \frac{1-\|z\|^2}{|1-z^\dagger \xi|^2} \right)^m.
\end{equation}
The densities \eqref{Cauchy-Bergman} define a sub-family of probability distributions on spheres in complex vector spaces, parametrized by points $z$ in the Bergman balls. We name this family the Cauchy-Bergman family and denote by ${\cal CB}(z)$.

The family ${\cal CB}(z)$ is invariant w.r. to automorhisms \eqref{automob} of the Bergman ball. More precisely, if $\Xi \sim  {\cal CB}(z)$ and $h$ is a biloholomorphic automorphism of $\mathbb{B}_{\mathbb{C}}^m$, then $h(\Xi) \sim {\cal CB}(h(z))$. 


We further investigate the energy level for this canonical family. Recall that the normalizing constant is given by \eqref{normalizing_constant_Bergman} and differentiate
\begin{equation}
\label{energy_level_Bergman}
const = - \frac{\partial}{\partial \lambda} \left. \log Z(\lambda,\|z\|^2) \right|_{\lambda=m} = - \log(1-\|z\|^2) - \frac{\frac{\partial}{\partial \lambda} \left. _2 F_1(\lambda,\lambda;m;\|z\|^2)\right|_{\lambda=m}}{_2 F_1(m,m;m;\|z\|^2)}.
\end{equation}
We can easily evaluate
$$
\frac{\partial}{\partial \lambda} \left. _2 F_1(\lambda,\lambda;m;\|z\|^2) \right|_{\lambda=m} = -2 (1-\|z\|^2)^{-m} \log (1-\|z\|^2).
$$
Substituting this equation and \eqref{hyp_geom_eval} into \eqref{energy_level_Bergman} we find that
$$
const = \mathbb{E} [\beta_\xi(z)] = - \log(1 - \|z\|^2) + 2 \log(1-\|z\|^2) = \log(1-\|z\|^2).
$$
This unveils the remarkable fact that the canonical energy is precisely the negative of the K\" ahler potential in the Bergman ball. Notice that it does not depend on the dimension $m$.

\begin{proposition}
MaxEnt distributions on the sphere in complex vector space with the fixed mean Busemann distance equal to the negative K\" ahler potential $\log(1-\|z\|^2)$ are defined by densities \eqref{Cauchy-Bergman}.
The corresponding value of the inverse temperature is equal to the complex dimension $m$.
\end{proposition}

\begin{remark}
The Fisher information matrix for the family ${\cal CB}(z)$ is given by
$$
F_{CB}(z) = \frac{m}{(1-\|z\|^2)^2} [(1-\|z\|^2) I + z z^\dagger].
$$
For the proof we refer to \cite{Jacimovic1}. 

We immediately notice that the above expression defines the Bergman metric. Hence, the information manifold ${\cal BC}(z)$ is isomorphic to the Bergman ball. 
\end{remark} 
 
\begin{remark}
Notice that the Cauchy-Bergman sub-family includes the uniform distribution for the reference point $z=0$. Delta distributions appear as limits when $\|z\| \to 1$.
\end{remark}

\section{Geometry and temperature on the Cauchy-Bergman statistical manifold}

Our next goal is to demonstrate potential relevance of the MaxEnt families derived in the previous Section for modeling. The exposition at this point may diverge, as we have two manifolds with the underlying conformal (Poincar\' e balls) or holomorphic geometry (Bergman balls). Both geometries are rich enough to underlie powerful models, but the difference between them is substantial. One apparent contrast is that the Bergman metric is anisotropic, while the Poincar\' e's is isotropic. This difference manifests in many computational and modeling aspects. In order to keep the exposition reasonably concise in the remainder of this study we focus on Bergman balls and the Cauchy-Bergman family on the sphere $\partial \mathbb{B}_{\mathbb{C}}^m$. 

This Section is divided into three subsections. In the first subsection we introduce the Busemann free energy as the crucial functional for modeling cognition. As expected, the family \eqref{MaxEnt_Busemann_constraint_Bergman} is recovered as Boltzmann-Gibbs distributions for this free energy functional. In the second subsection we focus on the Cauchy-Bergman sub-family and introduce the reproducing kernel as a measure of similarity between two C-B distributions. In the third subsection we introduce the RKHS associated with this reproducing kernel.

\subsection{Busemann variational free energy in the Bergman ball}

As pointed out in Introduction, the Busemann function acts a directional potential field in the ball. Hence, we introduce the Busemann variational free energy in the following way
\begin{equation}
\label{free_energy}
{\cal F}(p) = \int_{\partial \mathbb{B}_\mathbb{C}^m} \beta_\xi(z) p(\xi) d \sigma(\xi) - \frac{1}{\lambda} H(p),
\end{equation}
where $H(p)$ is the Shannon entropy \eqref{entropy} and $\beta_\xi(\cdot)$ is the Busemann function \eqref{Busemann_Bergman}.

The Euler-Lagrange condition for this functional yields
$$
\frac{\partial L}{\partial p} = \beta_\xi(z) + \frac{1}{\lambda} (\log p(\xi) + 1) + \mu = 0.
$$
Exponentiating this equation and normalizing the probability density we easily recover densities \eqref{MaxEnt_Busemann_constraint_Bergman} as minimizers of \eqref{free_energy}. 
Plugging these densities into \eqref{free_energy} we find the minimal value of the free energy:
$$
{\cal F}_{min} = - \frac{1}{\lambda} \log Z(\lambda,\|z\|^2).
$$
Hence, in line with the general theory of statistical-mechanical equilibrium, the equilibrium free energy in the Bergman ball is determined by the partition function of \eqref{MaxEnt_Busemann_constraint_Bergman}.

The perfect equilibrium state harmony is achieved for the case when $\lambda = m$: 
$$
Z(m,\|z\|^2) = (1-\|z\|^2)^{-m}.
$$
 Hence, in this case the equilibrium energy level is 
 $$
 {\cal F}_{min} = \log(1-\|z\|^2) = - K(z,\bar z),
 $$
  where $K(\cdot,\cdot)$ denotes the standard Bergman kernel \eqref{Bergman_kernel}.
  
This unveils a remarkable insight.

\begin{corollary}
At the canonical scale when the inverse temperature equals the complex dimension, the entire geometry of the Bergman ball arises as a consequence of the minimization of the free energy.
\end{corollary}

\subsection{Berezin kernel as a measure of overlap between Cauchy-Bergman distributions} \label{Poisson-Szego-sub-sect}

As shown in the previous Section, the information manifold ${\cal CB}(z)$ of the Cauchy-Bergman distributions is the Bergman ball. 

On the other hand, densities \eqref{Cauchy-Bergman} are boundary integral kernels for a point $\xi \in \partial \mathbb{B}_{\mathbb{C}}^m$. These densities have a real-analytic continuations from the boundary sphere to the interior ball:
\begin{equation}
\label{Poisson_Szego}
B(z,\zeta) = \left( \frac{(1-\|z\|^2)(1-\|\zeta\|^2)}{|1-z^\dagger \zeta|^2} \right)^m, \quad z,\zeta \in {\mathbb B}_{\mathbb C}^m. 
\end{equation}

The formula \eqref{Poisson_Szego} defines real-valued kernels on the Bergman balls, which are usually referred to as Berezin (or Poisson-Bergman) kernels \cite{Krantz2} and serve as integral kernels for the Berezin transform \cite{Zhu,AIM}. 

In order to introduce the notion of similarity (overlap) between two probability measures on ${\cal CB}(z)$ we recall the Bhattacharyya coefficient between two densities \cite{Bhattacharyya}:
\begin{equation}
\label{Bhatt_coef}
BhattC(z_1,z_2) = \int_{\partial \mathbb{B}^m_{\mathbb C}} \sqrt{p(\xi \, | \, z_1)p(\xi \, | \, z_2)} d \sigma(\xi), \; \mbox{ where } z_1,z_2 \in \mathbb{B}^m_{\mathbb{C}}.
\end{equation}
For the Cauchy-Bergman densities \eqref{Cauchy-Bergman} the integrand in \eqref{Bhatt_coef} can be evaluated as:
$$
\sqrt{p(\xi \, | \, z_1)p(\xi \, | \, z_2)} = c_m \frac{(1-\|z_1\|^2)^{m/2}(1-\|z_2\|^2)^{m/2}}{|1-z_1^\dagger \xi|^m|1-z_2^\dagger \xi|^m}.
$$
Plugging this into \eqref{Bhatt_coef} and using the Forelli-Rudin summation formula \cite{FR}:
$$
\int_{\partial \mathbb{B}^m_{\mathbb C}} \frac{1}{|1-z_1^\dagger \xi|^m|1-z_2^\dagger \xi|^m} d \sigma(\xi) = \frac{1}{c_m |1-z_1^\dagger z_2|^m}
$$
we can evaluate the spherical integral \eqref{Bhatt_coef} to get:
\begin{equation}
\label{Bhatt_coef_Poisson_Szego}
BhattC(z_1,z_2) = \left( \frac{(1-\|z_1\|^2)(1-\|z_2\|^2)}{|1-z_1^\dagger z_2|^2} \right)^{m/2} = B(z_1,z_2)
\end{equation}
where $B(\cdot,\cdot)$ is the Berezin kernel \eqref{Poisson_Szego}.

Therefore, kernel \eqref{Poisson_Szego} encodes similarity between two probability measures on the manifold ${\cal CB}(z)$. 

\begin{remark}
Real-valued kernel \eqref{Poisson_Szego} is the squared modulus of the complex-valued normalized Bergman kernel in the ball \cite{Zhu}:
\begin{equation}
\label{normal_Bergman_kernel}
K_{nB}(z,\zeta) = \frac{K(z,\zeta)}{\sqrt{K(z,z)}\sqrt{K(\zeta,\zeta)}} = \left( \frac{(1-\|z\|^2)(1-\|\zeta\|^2)}{1-z^\dagger \zeta} \right)^m, \quad z,\zeta \in {\mathbb B}_{\mathbb C}^m.
\end{equation}
Kernels \eqref{normal_Bergman_kernel} measure an overlap between two coherent quantum states \cite{Perelomov,Perelomov2}. 
Indeed, \eqref{normal_Bergman_kernel} can be written as scalar product
$$
K_{nB}(z,\zeta) = \phi(\zeta)^\dagger \phi(z) = \sum_{\alpha} \phi^{(\alpha)}(\zeta) \overline {\phi^{(\alpha)}(z)},
$$
with functions
$$
\phi^{(\alpha)}(z) = (1-\|z\|^2)^{\frac{m+1}{2}} \sqrt{\frac{(|\alpha| + m)!}{\alpha! \cdot m!}} \cdot z_1^{\alpha_1} z_2^{\alpha_2} \cdots z_m^{\alpha_m}
$$
where $\alpha \in {\mathbb N}_0^m$ is multi-index with $|\alpha| = \sum_{j=1}^m \alpha_j$ and $\alpha ! = \alpha_1! \alpha_2! \cdots \alpha_m!$.

Their squared modulus \eqref{Poisson_Szego} is interpreted as transition probability between states $z$ and $\zeta$.
\end{remark}

\subsection{RKHS associated with the Berezin reproducing kernel} \label{RKHS-sub-sect}

Since the Berezin kernel is symmetric and positive definite, the Moore-Aronszajn theorem \cite{HSS} states that there exists a unique Reproducing Kernel Hilbert Space (RKHS) where it is a reproducing kernel. The structure of this RKHS is well known. It consists of square integrable functions which are annihiliated by the Laplace-Beltrami operator on $\mathbb{B}_{\mathbb{C}}^m$:
$$
\triangle_M = 4 (1-\|z\|^2) \sum (\delta_{jk} - z_j \bar z_k) \frac{\partial^2}{\partial z_j \partial z_k}.
$$ 
Functions annihiliated by the above operator are referred to as $M$-harmonic functions. We will denote this Hilbert space by ${\cal H}_M$. They posses boundary limits almost everywhere on the sphere $\partial \mathbb{B}_{\mathbb{C}}^m$.

The inner product that satisfies the reproducing property $\langle f,K(\cdot,w) \rangle_{{\cal H}_M} = f(w)$ is defined via boundary extensions of the functions to the sphere:
$$
\int_{\partial \mathbb{B}^m_{\mathbb C}} f^*(\xi) \overline{g^*(\xi)} d \sigma(\xi),
$$
where $f^*(\xi) = \lim_{r \to 1} f(r \xi)$ denotes the boundary limit function.

\section{Holomorphic-information-geometric model of cognition and decision-making}

K\" ahler geometry of Bergman balls and favorable geometric properties of MaxEnt distributions offer an expressive and tractable framework for modeling cognitive processes. This geometry can capture various counterintuitive phenomena and encode latent hierarchies, uncertainties, or contextual ambiguities. The RKHS theory in Bergman balls provides a strong theoretical background.

In the remainder of this paper, we focus on one possible interpretation related with personal beliefs within the broader sociological landscape. For the sake of elegance and simplicity, the underlying mathematical framework is restricted to equilibrium states at the canonical energy level.

\subsection{The model: beliefs, states of the world and mind}

We conceive the boundary sphere $\partial \mathbb{B}^m_{\mathbb{C}}$ as an infinite horizon composed of pure states, where each point represents an "idealistic (platonic) image of the world". Roughly, this can be conceptualized in the following way. There is a given long (possibly infinite) sequence of questions, each of them allowing three possible responses: "agree", "disagree", or "neutral/indifferent". Points on the sphere correspond to a unique trinary sequence of answers. In such a setup the notions of similar, antipodal and orthogonal images of the world are transparent.

Accordingly, the interior ball $\mathbb{B}^m_{\mathbb{C}}$ defines the space of "coherent states of the mind". In other words, points in the ball encode individual beliefs with respect to images of the world. The center of the ball corresponds to complete ambivalence, where all outcomes are equally likely. On the other extreme, points with moduli close to one correspond to highly deterministic states of the mind. Persons with such beliefs have nearly deterministic answers to the questions. Notice that this includes not only radicals, but also moderate or apathetic persons. The last group would respond "neutral/indifferent" to the majority of questions. 

With each coherent state of the mind corresponding to $z \in \mathbb{B}^m_{\mathbb{C}}$ we associate the probability density \eqref{Cauchy-Bergman} parametrized by this point. When an individual in the state $z$ is asked a question from the list, they sample a random point on the sphere from \eqref{Cauchy-Bergman}. Since points on the sphere are associated with the prescribed list of answers, this sample determines the response. Similarity between coherent states of the mind is defined by the kernel \eqref{Poisson_Szego}. 

Finally, $M$-harmonic functions represent collective values. Each function in this space represents a collective evaluative landscape or a socio-political narrative spanning the entire society. When evaluated at a specific point in the Bergman ball, this function computes the degree of cognitive resonance, or alignment between the individual state of the mind and the collective value map. 


Within this interpretation, the norm of the function $\| f\|_{{\cal H}_M}$ serves as an indicator of social polarization or ideological intensity. A low-norm function represents a smooth evaluative landscape where value transitions across the cognitive space are gradual, reflecting low friction within society. Conversely, a high-norm function indicates a  sharp, volatile landscape characterized by steep gradients. In the socio-political context, these high-norm functions model highly polarizing and divisive issues that fragment the society. The pointwise constraints $f(u_i) = c_i$ in RKHS encode empirical facts that bind the global value map to known real-world data. 
Function $f$ propagates these constraints across the interior of the ball, implicitly dictating how the surrounding, unmeasured individuals will rationally resonate with the collective narrative based on their proximity to the anchor points.

Now, suppose that the collective landscape is described by a function $f$ in RKHS. Since it is an $M$-harmonic function, it has the radial boundary limit $f^*(\xi) = \lim_{r \to 1} f(r \xi)$. The function $f^*$ on the boundary sphere defines the value (acceptability) of certain platonic ideologies within the society.

On the other hand, if we know the density values across all ideologies, the collective landscape can be recovered via the boundary integral:
$$
f(z) = \int_{\partial \mathbb{B}_{\mathbb{C}}^m} p(\xi \, | \, z) f^*(\xi) d \sigma(\xi)
$$
where $p(\xi \, | \, z)$ is the Cauchy-Bergman density \eqref{Cauchy-Bergman}.

\subsection{The principle of minimal collective cognitive effort} \label{Min-effort-sub-sect}

We further postulate that a society tends to occupy the smoothest cognitive landscape under constraints at anchor points. Mathematically this means that the collective landscape will be described by the minimal norm function in RKHS satisfying given constraints.

For instance, consider a society with a single anchor point. We may conceive that this point corresponds to a single narrative that everyone accepts. For instance, it could be the consensual recognition that all of them are citizen of the same state, while all other standpoints remain completely open. 

The collective landscape for such a society is described by the minimal norm function in ${\cal H}_M$ under this single constraint:
\begin{equation}
\label{min_norm_RKHS}
\min_f \| f \|_{{\cal H}_M}, \; \mbox{ subject to  } \; f(u)=c.
\end{equation}
The representer theorem \cite{HSS} states that the solution to the above optimization problem is
$$
\hat f(z) = \alpha B(z,u)
$$
where $\alpha = c/B(u,u)$ and $B(\cdot,\cdot)$ is the Berezin kernel \eqref{Poisson_Szego}. 

Hence, the minimal norm function is:
\begin{equation}
\label{min_norm_function}
\hat f(z) = c \frac{B(z,u)}{B(u,u)} = c B(z,u)
\end{equation}
with the norm
$$
\| \hat f \|_{{\cal H}_M} = \frac{c^2}{B(u,u)} = c^2,
$$
where we used that $B(u,u)=1$. Hence, the minimal norm does not depend on the point $u \in \mathbb{B}^m_{\mathbb{C}}$. In other words, the minimal norm value is the same for the case when the consensual narrative is vague or self-evident (anchor point close to the center of the ball) and when the narrative is narrow and specific (a point far from the center). 

The norm can also be evaluated from the limit of $\hat f$ to the boundary sphere:
$$
\| \hat f \|^2 = \int_{\partial \mathbb{B}_{\mathbb{C}}^m} |\hat f^*(\xi)|^2 d \sigma(\xi).
$$
This unveils the duality between the energy (norm) in RKHS and entropy in the latent states space. The minimal norm is the smoothest (least energy) function which satisfies the constraint $f(u)=c$ . The solution \eqref{min_norm_function} shows it is the Berezin kernel. This kernel has the radial limit given by the Cauchy-Bergman density centered at $u \in {\mathbb B}^m$. 

In general, there is more than one anchor point in a society. Accordingly, the interpolation problem is written as
$$
\min_f \| f \|_{{\cal H}_M}, \; \mbox{ subject to  } \; f(u_1)=c_1,\dots,f(u_k)=c_k.
$$
The representer theorem \cite{HSS} states that the solution is the linear combination of the Berezin kernels (coherent states)
$$
\hat f(z) = \sum_{j=1}^k \alpha_j B(z,u_j).
$$
Coefficients $\alpha_j$ are solutions of the system of linear equations:
$$
B \bm{\alpha} = \textbf{c},
$$
where $B_{ij} = B(u_i,u_j)$ is the Grammian matrix and vectors $\bm{\alpha} = (\alpha_1,\dots,\alpha_k), \, \textbf{c} = (c_1,\dots,c_k).$
Notice that the Grammian $B$ is regular if points $u_1,\dots,u_k$ are different.

The minimal norm value equals $\|\hat f\|^2_{{\cal H}_M} = {\textbf c}^\dagger B^{-1} {\textbf c}$.

Therefore, the minimal norm solution is a linear combination of coherent states (Berezin kernels). 

\subsection{Quantum-information-theoretic analogies}

Obviously, there is a strong analogy between our model and quantum-theoretic framework of $SU(m,1)$-coherent states in mathematical physics \cite{Perelomov}. Under this interpretation, the boundary sphere $\partial \mathbb{B}^m_{\mathbb{C}}$ acts as the space of pure quantum states, whereas the interior ball $\mathbb{B}_{\mathbb{C}}^{m}$ represents the mixed states space of coexisting, superimposed beliefs. The vacuum state is the origin of the ball, corresponding to the completely ambivalent state of the mind. The act of answering a question constitutes a quantum-like measurement, triggering a collapse from a mixed state $z$ onto a pure state defined by a boundary point.

The crucial modeling choice was made in subsection \ref{Poisson-Szego-sub-sect} where we used the Bhattacharyya coefficient as a measure of similarity between coherent states of the mind. Such a choice led to the real-valued kernel \eqref{Poisson_Szego} and the space of real-analytic $M$-harmonic functions. This choice implies that minimal norm functions in subsection \ref{Min-effort-sub-sect} are real linear combinations of real-valued kernels. Such functions can be interpreted as densities over the space of coherent states. 

Using a real-valued reproducing kernel rather than the complex-valued one imposes some modeling limitations, but has certain computational advantages. In particular, real-valued kernels are convenient for classification and regression using the kernel support vector machines (SVM). Furthermore, they add conceptual clarity as Berezin kernels are precisely transition probabilities between two coherent states \cite{BZ}. Therefore, the Born's rule is already incorporated in the model. Finally, the Calabi's diastasis between two coherent states $u$ and $z$ is expressed as the negative logarithm of the this transition probability: $
D_C = - \log K(z,u),$ where $K(\cdot,\cdot)$ is given by \eqref{Poisson_Szego}. Suppose that we have $k$ coherent states $z_1,\dots,z_k$. Then the minimum of the function:
$$
F(u) = - \sum_{i=1}^k D_C(u,z_i)
$$
is precisely the holomorphic barycenter of points $z_1,\dots,z_k$ as introduced in \cite{JK,Jacimovic}. We also refer to the recent preprint \cite{CY} which reports novel insight into information-theoretic contest of the Bergman geometry. Recalling the interpretation of the Calabi's diastasis as the K-L divergence, this has clear interpretation: the holomorphic barycenter is the information filter which distills states $z_1,\dots,z_k$ with minimal surprisal (in terms of active inference \cite{Friston}. 

On the downside, the model with real-valued kernel has does not have sufficient expressive power to capture directed influences within a society and some quantum-like effects in cognition. This is discussed in the next Section.

\section{Conclusion and outlook}

In the first part of paper we presented families of MaxEnt probability distributions on spheres with the fixed value of expected Busemann function to a given reference point in the interior ball. We demonstrated that this hyperbolic-geometric energetic constraint yields a highly tractable family of Boltzmann-Gibbs distributions. 

In particular, wrapped Cauchy distributions (Poisson kernels), the most natural and computationally tractable family of densities on the circle, arise as canonical Boltzmann-Gibbs distributions on the circle under the Busemann energy constraint in the interior Poincar\' e disc. 

The Poincar\' e disc is the minimal model of hyperbolic geometry, a manifold of complex dimension one. In higher dimensions there are two non-equivalent models of hyperbolic balls. Both of them reduce to the Poincar\' e disc for real dimension $d=2$, or complex dimension $m=1$. We analyzed MaxEnt families for both cases and emphasized the differences. For instance, the canonical equilibrium state in the Bergman ball induces a perfect geometric harmony with free energy equal to the negative K\" ahler potential. On the other hand, free energy levels at the canonical temperature in Poincar\' e balls depend on the dimension and contain a radius-dependent polynomial term.

In the second part of the paper we focused on Bergman balls to present geometric model of collective cognition. An analogous model can be developed in the Poincar\' e ball as well. This is the matter of modeling choice, as both manifolds are attractive for sophisticated geometric modeling. However, the Bergman ball provides more expressive framework, encompassing both hyperbolic and quantum representations. Unlike Poincar\' e's, Bergman balls are anisotropic, which boosts their expressive power.

\subsection{Limitations of the presented model with an outlook on possible upgrades and applications}

In the present paper we presented the basic model, which can be built upon for further upgrades. In its current form, the model has a limited capacity to capture certain effects. However, it is very flexible. Below, we list some of the limitations and briefly indicate the ways of upgrading the model to take them into account.
 
\begin{itemize}

\item{Cognitive barriers and identity-protective cognition.} 

People are averse to accepting uncomfortable facts. They unconsciously credit or dismiss evidence to protect their status, belonging, and bonds within a social group. This effect can be incorporated in our model by introducing state dependent constraints in RKHS. This would further lead to the state-dependent matrix $K$ in subsection \ref{Min-effort-sub-sect} instead of the constant one. Consequently, the linear system $K \bm{\alpha} = \textbf{c}$ would turn into a nonlinear one.

\item{Actions: active inference.} 

The Busemann free energy \eqref{free_energy} can be a cornerstone of geometric active inference inspired by the seminal framework of Karl Friston \cite{Friston}. In our geometric setup, this would include the K\" ahler potential as an equilibrium energy level. However, in order to implement the full active inference framework, one needs to incorporate actions into the free energy functional and to differentiate it w.r. to both perception and action. It can be realized in an elegant and tractable way, obtaining closed-form expressions. This generalization towards spherical (or quantum) active inference is a focus of ongoing research.

\item{Asymmetric influences and flows.} 

One exciting opportunity is modeling cognitive or sociological evolution as flows in RKHS's or in Bergman balls. Our framework could naturally incorporate influences between individuals or groups. However, real-valued reproducing kernels can not model asymmetric influences. In order to model this, one should utilize complex-valued normalized Bergman kernel \eqref{normal_Bergman_kernel}.

\item{Quantum-like effect: measurement changes the system.} 

Since the foundation of quantum mechanics the problem of measurement puzzled scientists and philosophers. This widely known paradox arises from the fact that an act of observation changes the system being measured. Analogous effects have been experimentally observed and documented in human cognition \cite{Khrennikov,BB}. For instance, the quantum-like nature of human cognition manifests when the people are asked two questions in different orders. The probability of responding positively to a question A depends significantly on if this question is asked before or after question B. Therefore, asking a question changes someone's state of the mind. In the same manner, the act of voting changes a society. There are several ways to incorporate this in our model. One interesting idea is to explain it as cyclic evolution and appearance of quantum geometric phase (holonomy) \cite{Berry}. Bergman balls are very natural setup to model this effect. However, this might be conceptually demanding. A more intuitive way is the mapping between pure states on the sphere and constraints in the RKHS, as in the first point in this list.

\item{Interference of the cognitive states.}

Another quantum-like effect in cognition is interference, which can be constructive or destructive. As already pointed out, real-valued kernel entail real solutions in \ref{Min-effort-sub-sect}. In such a setup, interference can not be modeled directly in RKHS. A complex-valued reproducing kernel is associated with the RKHS of holomorphic functions. In such a setup, functions in RKHS are analogues of wave functions and interference appears naturally. Although this allows for modeling quantum-like effects, it is associated with potentially severe computational difficulties in large models. Sampling is a good example. With real linear combinations, we can efficiently sample points on the sphere from mixtures of the Cauchy-Bergman distributions. If the coefficients are complex one must apply expensive acceptance-rejection schemes.

\end{itemize}

Overall, the model can be significantly upgraded by considering both kernels \eqref{Poisson_Szego} and \eqref{normal_Bergman_kernel} and their associated RKHS's for modeling various effects.

The bottomline of this study is that hyperbolic (both Poincar\' e and Bergman) balls and associated theory of reproducing kernels provide a powerful framework for compact ML models and cognitive AI architectures with potential applications in sociology, psychology and other fields. The potential of spherical and hyperbolic representations in deep learning has been widely recognized in the recent decade. For instance, the cosine metric and representations of words by unit vectors are routinely used in most of LLM's. As a simple preliminary insight, spherical representations are natural choice whenever entities or concepts have unique antipodes. We also point out that the generalized Kuramoto models on spheres \cite{LMS} can serve as a computational tool in spherical and hyperbolic ML architectures. 


Bergman balls are particularly attractive manifolds for modeling subtle concepts and complicated (often counterintuitive) behavioral patterns. The underlying deep theory of bounded symmetric domains in complex vector spaces \cite{Zhu,Satake} has attracted attention of mathematicians since for one century. In the second half of the 20th century this theory found its interpretation in the formalism of coherent states in quantum physics \cite{Perelomov}. Nowadays, this abstract mathematical theory is very well established, while continuously evolving with new insights and beautiful theorems. With the explosive progress of cognitive AI, this theory in its information-theoretic disguise can provide a solid foundation for the deployment of compact and sophisticated models.

\section*{Appendix: Calculation of derivatives of hypergeometric functions}

For the sake of completeness we derive formulae on derivatives of hypergeometric functions used in the main text.

\begin{proposition}
\label{hyp_diff}
Let $m$ be an integer. Then 
$$
\frac{\partial}{\partial \lambda} \, \left. _2 F_1(\lambda,\lambda;m,r^2) \right|_{\lambda=m} = 2 \frac{\log(1-r^2)}{(1-r^2)^m},
$$
where $_2 F_1(\cdot,\cdot;\cdot;\cdot)$ are Gauss hypergeometric series.
\end{proposition}

\begin{proof}

Consider the hypergeometric function $_2 F_1(\lambda_1,\lambda_2;m;r^2).$

Using the chain rule, we have:
$$
\frac{\partial}{\partial \lambda} \, \left. _2 F_1(\lambda,\lambda;m;r^2) \right |_{\lambda=m} = \frac{\partial}{\partial \lambda_1} \, \left. _2 F_1(\lambda_1,m;m;r^2) \right |_{\lambda_1=m} + \frac{\partial}{\partial \lambda_2} \, \left. _2 F_1(m,\lambda_2;m;r^2) \right |_{\lambda_2=m}.
$$
Since the first two variables of the Gauss hypergeometric function are symmetric, their partial derivatives are identical when $\lambda_1=\lambda_2$. Using this fact and the above equality we have that:
$$
\frac{\partial}{\partial \lambda} \, \left. _2 F_1(\lambda,\lambda;m;r^2) \right |_{\lambda=m} = 2 \frac{\partial}{\partial \lambda_1} \, \left. _2 F_1(\lambda_1,m;m;r^2) \right |_{\lambda_1=m}.
$$
We have that $_2 F_1(\lambda_1,m;m;r^2) = (1-r^2)^{-\lambda_1}$. Substitute this back into the derivative equality to get:
$$
\frac{\partial}{\partial \lambda} \, \left. _2 F_1(\lambda,\lambda;m;r^2) \right |_{\lambda=m} = 2 \frac{\partial}{\partial \lambda_1} \left. (1-r^2)^{-\lambda_1} \right|_{\lambda_1=m}.
$$
Now the right hand side of the above equality can be simply differentiate to get
$$
\frac{\partial}{\partial \lambda} \, \left. _2 F_1(\lambda,\lambda;m;r^2) \right |_{\lambda=m} = \left. -2(1-r^2)^{-\lambda_1} \log(1-r^2) \right|_{\lambda_1=m} = - 2 \frac{\log(1-r^2)}{(1-r^2)^m}.
$$

\end{proof}

\begin{proposition}
Let $d$ be an integer. Then
$$
\frac{\partial}{\partial \lambda} \left. _2 F_1(\lambda,\lambda-\frac{d-2}{2};\frac{d}{2};r^2) \right|_{\lambda=d-1} = - \frac{2 \log(1-r^2)}{(1-r^2)^{d-1}} + \frac{P_{d-3}(r^2)}{(1-r^2)^{d-1}}.
$$ 
where 
\begin{equation}
\label{P_d-3}
P_{d-3}(r^2) = \sum_{j=1}^{d-2} \frac{1}{j} \sum_{k=1}^j (-1)^k \binom{j}{k} r^{2k}.
\end{equation}
\end{proposition}

\begin{proof}
Fix $d$ and $r^2$ and consider the two variable function $g(a,b) = _2 F_1(a,b;d/2;r^2)$. Then applying the chain rule to $g(\lambda,\lambda-(d-2)/2)$ we have
$$
\frac{\partial g}{\partial \lambda} = \frac{\partial g}{\partial a} + \frac{\partial g}{\partial b}.
$$
In order to compute the derivative w. r. to $a$ set $b=d-1$ and find that $g(a,d/2) = (1-r^2)^{-a}$.

Differentiating this w. r. to $a$:
$$
\frac{\partial}{\partial a} [(1-r^2)^{-a}] = - (1-r^2)^{-a} \log(1-r^2).
$$
Then the derivative at the point $a=d-1$ reads
\begin{equation}
\label{derivative_a}
\left. \frac{\partial g}{\partial a} \right|_{\lambda=d-1} = - \frac{\log(1-r^2)}{(1-r^2)^{d-1}}.
\end{equation}

Now, we pass to the derivative w. r. to $b$. We expand the hypergeometric function into series and differentiate to get
$$
\left. \frac{\partial g}{\partial b} \right|_{b=d/2} = \sum_{n=1}^\infty \frac{(d-1)_n \left(\frac{d}{2}\right)_n}{\left(\frac{d}{2}\right)_n n!} \left[\psi \left(\frac{d}{2} + n\right) - \psi \left(\frac{d}{2}\right) \right] r^{2n},
$$ 
where $\left(d/2\right)_n$ denotes the Pochamer symbol and $\psi$ is the digamma function. 
Canceling Pochamer symbols and writing digamma function as sums, this is rearranged as 
$$
\left. \frac{\partial g}{\partial b} \right|_{b=d/2} = \sum_{n=1}^\infty \frac{(d-1)_n}{n!} \left( \sum_{k=0}^{n-1} \frac{1}{d/2+k} \right) r^{2n}.
$$
Changing the order of summation and taking the sum for the index $n$:
\begin{equation}
\label{derivative_b}
\left. \frac{\partial g}{\partial b} \right|_{b=d/2} = - \frac{\log(1-r^2)}{(1-r^2)^{d-1}} + \frac{1}{(1-r^2)^{d-1}} \sum_{j=1}^{d-2} \frac{(1-r^2)^j -1}{j}.
\end{equation}
Summing \eqref{derivative_a} and \eqref{derivative_b}:
$$
\frac{\partial}{\partial \lambda} \, \left. _2 F_1(\lambda,\lambda-\frac{d-2}{2};\frac{d}{2};r^2) \right|_{\lambda=d-1} = -2 \frac{\log(1-r^2)}{(1-r^2)^{d-1}} + \frac{1}{(1-r^2)^{d-1}} \sum_{j=1}^{d-2} \frac{(1-r^2)^j - 1}{j}.
$$ 
Applying the binom formula, we find that the sum in the second term equals $P_{d-3}(r^2)$ where $P_{d-3}$ is defined by \eqref{P_d-3}. This completes the derivation.

\end{proof}

\section*{Declaration of Generative AI in Scientific Writing}

The author acknowledges the assistance of an AI tool developed by Google in exploring and formatting the mathematical derivations and improving the manuscript’s language. Every step of the mathematical proofs was independently checked and validated by the author, who remains fully responsible for the final manuscript.

\end{document}